\documentclass[letterpaper, 10 pt, conference]{ieeeconf}
\IEEEoverridecommandlockouts
\usepackage[english]{babel}
\usepackage{theorem}
\theoremheaderfont{\itshape\bfseries}
{\theorembodyfont{\itshape}
	\newtheorem{assumption}{\textbf{Assumption}}
	\newtheorem{lemma}{\textbf{Lemma}}
	\newtheorem{definition}{\textbf{Definition}}
	\newtheorem{theorem}{\textbf{Theorem}}
	
	\newtheorem{remark}{\textbf{Remark}}
	
	\newtheorem{problem}{\textbf{Problem}}
}
\usepackage{graphicx}
\usepackage{newtxtext,newtxmath}
\usepackage{amsmath}
\usepackage{amsfonts}
\usepackage{mathrsfs}
\usepackage{xcolor}
\usepackage{graphicx}

\newcommand{\T}{^{\mbox{\tiny T}}}
\newcommand{\R}{\mathbb{R}}

\newenvironment{proof}[1][Proof]%
{\par\noindent\textit{#1:\ }}%
{\hspace*{\fill} \rule{6pt}{6pt}}
\newenvironment{proof*}[1][Proof]%
{\par\noindent\textit{#1:\ }}{}

\DeclareMathOperator{\diag}{diag}

\DeclareMathOperator{\rank}{rank}

\usepackage{tikz}
\usetikzlibrary{shapes,calc,arrows,patterns,decorations.pathmorphing
	,decorations.markings}
\usetikzlibrary{arrows.meta}

\newenvironment{system}[1]%
{\setlength{\arraycolsep}{0.5mm}\left\{ \; \begin{array}{#1}}%
	{\end{array} \right.}
\newenvironment{system*}[1]%
{\setlength{\arraycolsep}{0.5mm} \begin{array}{#1}}%
	{\end{array}}

\title{\LARGE \textbf{Scale-free Protocol Design for Output and Regulated Output Synchronization of Heterogeneous Multi-agent Systems}}

\author{Donya Nojavanzadeh, Zhenwei Liu, Ali Saberi and Anton A. Stoorvogel%
	\thanks{Donya Nojavanzadeh is with School of Electrical Engineering
		and Computer Science, Washington State University, Pullman, WA
		99164, USA {\tt\small donya.nojavanzadeh@wsu.edu}}
		\thanks{Zhenwei Liu is with College of Information Science and
		Engineering, Northeastern University, Shenyang 110819,
		P. R. China {\tt\small jzlzwsy@gmail.com}}
			\thanks{Ali Saberi is with
		School of Electrical Engineering and Computer Science, Washington
		State University, Pullman, WA 99164, USA {\tt\small
			saberi@eecs.wsu.edu}} 
			\thanks{Anton A. Stoorvogel is with
			Department of Electrical Engineering, Mathematics and Computer
			Science, University of Twente, P.O. Box 217, Enschede, The
			Netherlands {\tt\small A.A.Stoorvogel@utwente.nl}}
	}

\begin{document}

	\maketitle
\thispagestyle{empty}
\pagestyle{empty}

	\begin{abstract}
 In this paper, we consider scalable output and regulated output synchronization problems for heterogeneous networks of right-invertible
 linear agents based on
 localized information exchange where in the case of regulated output synchronization, the reference trajectory is generated by a so-called exosystem. We assume that all the agents are introspective, meaning that they have access to their own local measurements.
We propose a scale-free linear protocol for each agent to achieve output and regulated output synchronizations. These protocols are designed solely based on agent models and they need no information about communication graph and the number of agents or other agent models information.
\end{abstract}

\section{Introduction}

Synchronization problem of multi-agent systems (MAS) has become a hot topic among researchers in recent years. Cooperative control of MAS is used in practical application such as robot network, autonomous vehicles, distributed sensor network, and others. The objective of synchronization is to secure an asymptotic agreement on a common state or output trajectory by local interaction among agents (see \cite{bai-arcak-wen,mesbahi-egerstedt,ren-book,wu-book} and references therein).

State synchronization inherently requires homogeneous networks. Most works have focused on state synchronization where each agent has access to a linear combination of its own state relative to that of the neighboring agents, which is called full-state coupling \cite{saber-murray3,saber-murray,saber-murray2,ren-atkins,ren-beard-atkins,tuna1}. A more realistic scenario which is partial-state coupling (i.e. agents share part of their information over the network) is studied in \cite{tuna2,li-duan-chen-huang,pogromsky-santoboni-nijmeijer,tuna3}.

For heterogeneous network it is more reasonable to consider output synchronization since the dimensions of states and their physical interpretation may be different. Introspective agents possess some knowledge about their own states. For heterogeneous MAS with non-introspective agents, it is well-known that one needs to regulate outputs of the agents to a priori given trajectory generated by a so-called exosystem (see
\cite{wieland-sepulchre-allgower, grip-saberi-stoorvogel3}). Other works on synchronization of MAS with non-introspective agents can be found in the literature as  \cite{grip-yang-saberi-stoorvogel-automatica,grip-saberi-stoorvogel}.

On the other hand, for MAS with introspective agents, one can achieve output and regulated output synchronization.
Most of the literature for heterogeneous MAS with introspective agents are based on modifying the agent dynamics via local feedback to achieve some form of homogeneity. There have been many results for synchronization of heterogeneous networks with introspective agents, see for instance \cite{kim-shim-seo,yang-saberi-stoorvogel-grip-journal,li-soh-xie-lewis-TAC2019,modares-lewis-kang-davoudi-TAC2018,qian-liu-feng-TAC2019,chen-auto2019}.

In this paper, we propose \textbf{scale-free} protocol design to solve output and regulated output synchronization problems for heterogeneous MASs with introspective right-invertible agents. Scale-free protocols are designed based on localized information exchange with neighbors and do not require any knowledge of the communication network except connectivity. The protocol design is scale-free, namely,
\renewcommand\labelitemi{{\boldmath$\bullet$}}
\begin{itemize}
	\item  The design is independent of the information about communication networks such as spectrum of the associated Laplacian matrix. That is to say, the universal dynamical protocols work for any communication network as long as it is connected.
	\item  The dynamic protocols are designed solely based on agent models and do not depend on communication network and the number of agents.
	\item The proposed protocols achieve output and regulated output synchronization for heterogeneous MAS with any number of agents and any communication network.
\end{itemize}

\subsection*{Notations and definitions}

Given a matrix $A\in \mathbb{R}^{m\times n}$, $A\T$ denotes its
conjugate transpose. A square matrix
$A$ is said to be Hurwitz stable if all its eigenvalues are in the
open left half complex plane. We denote by
$\diag\{A_1,\ldots, A_N \}$, a block-diagonal matrix with
$A_1,\ldots,A_N$ as its diagonal elements. $A\otimes B$ depicts the
Kronecker product between $A$ and $B$. $I_n$ denotes the
$n$-dimensional identity matrix and $0_n$ denotes $n\times n$ zero
matrix; sometimes we drop the subscript if the dimension is clear from
the context.

To describe the information flow among the agents we associate a \emph{weighted graph} $\mathcal{G}$ to the communication network. The weighted graph $\mathcal{G}$ is defined by a triple
$(\mathcal{V}, \mathcal{E}, \mathcal{A})$ where
$\mathcal{V}=\{1,\ldots, N\}$ is a node set, $\mathcal{E}$ is a set of
pairs of nodes indicating connections among nodes, and
$\mathcal{A}=[a_{ij}]\in \mathbb{R}^{N\times N}$ is the weighted adjacency matrix with non negative elements $a_{ij}$. Each pair in $\mathcal{E}$ is called an \emph{edge}, where
$a_{ij}>0$ denotes an edge $(j,i)\in \mathcal{E}$ from node $j$ to
node $i$ with weight $a_{ij}$. Moreover, $a_{ij}=0$ if there is no
edge from node $j$ to node $i$. We assume there are no self-loops,
i.e.\ we have $a_{ii}=0$. A \emph{path} from node $i_1$ to $i_k$ is a
sequence of nodes $\{i_1,\ldots, i_k\}$ such that
$(i_j, i_{j+1})\in \mathcal{E}$ for $j=1,\ldots, k-1$. A \emph{directed tree} is a subgraph (subset
of nodes and edges) in which every node has exactly one parent node except for one node, called the \emph{root}, which has no parent node. A \emph{directed spanning tree} is a subgraph which is
a directed tree containing all the nodes of the original graph. If a directed spanning tree exists, the root has a directed path to every other node in the tree \cite{royle-godsil}.  

For a weighted graph $\mathcal{G}$, the matrix
$L=[\ell_{ij}]$ with
\[
\ell_{ij}=
\begin{system}{cl}
\sum_{k=1}^{N} a_{ik}, & i=j,\\
-a_{ij}, & i\neq j,
\end{system}
\]
is called the \emph{Laplacian matrix} associated with the graph
$\mathcal{G}$. The Laplacian matrix $L$ has all its eigenvalues in the
closed right half plane and at least one eigenvalue at zero associated
with right eigenvector $\textbf{1}$ \cite{royle-godsil}. 

\section{Problem Formulation}

We will study a MAS consisting of $N$ non-identical linear agents:
\begin{equation}\label{hete_sys}
\begin{system*}{cl}
\dot{x}_i&=A_ix_i+B_iu_i,\\
y_i&=C_ix_i,
\end{system*}
\end{equation}
where $x_i\in\mathbb{R}^{n_i}$, $u_i\in\mathbb{R}^{m_i}$ and $y_i\in\mathbb{R}^p$ are the state,
input, output of agent $i$
for $i=1,\ldots, N$.

The agents are introspective, meaning that each agent has access to its own local information. Specifically each agent has access to part of its state
\begin{equation}\label{local}
	z_i=C_i^mx_i.
\end{equation}
where $z_i\in \mathbb{R}^{q_i}$.

The communication network provides agent $i$ with the following information which is a linear combination of its own output relative to that of other agents:

\begin{equation}\label{zeta1}
\zeta_i=\sum_{j=1}^{N}a_{ij}(y_i-y_j)
\end{equation}
where $a_{ij}>0$ and $a_{ii}=0$. The communication topology of the network can be described by a weighted and directed graph $\mathcal{G}$ with nodes corresponding to the agents in the network and the weight of
edges given by the coefficient $a_{ij}$. In terms of the coefficients of the associated Laplacian matrix $L$, $\zeta_i$ can be rewritten as
\begin{equation}\label{zeta}
\zeta_i= \sum_{j=1}^{N}\ell_{ij}y_j.
\end{equation}

In this paper, we also introduce a localized information exchange among agents, namely, each agent $i\in\{1,...,N\}$ has access to localized information, denoted by $\hat{\zeta}_i$, of the form
\begin{equation}\label{etahat}
	\hat{\zeta}_i=\sum_{j=1}^{N}a_{ij}(\eta_i-\eta_j)
\end{equation}
where $\eta_i$ is a variable produced internally by agent $i$ and to be defined later.

 In order to explicitly state our problem formulation we need the following definition. 
\begin{definition}\label{def1}
	Let $\mathbb{G}^N$ denote the set of directed graphs of $N$ agents which contain a directed spanning tree.
\end{definition}

Now we formulate the problem of \textbf{scalable} output synchronization for a heterogeneous MAS.

\begin{problem}\label{prob_sync}
	Consider a heterogeneous network of $N$ agents \eqref{hete_sys} with local information \eqref{local} satisfying Assumption \ref{ass2}. Let the associated network communication be given by \eqref{zeta}. Let $\mathbb{G}^N$ be the set of network graphs as defined in Definition \ref{def1}. The \textbf{scalable output synchronization problem based on localized information exchange} is to find, if possible, a linear dynamic controller for each agent $i \in\{1, \dots, N\}$, using only knowledge of agent models, i.e. $(C_i,A_i,B_i)$, of the form: 
	\begin{equation}\label{out_dyn}
	\begin{system}{cl}
	\dot{x}_{i,c}&=A_{i,c}x_{i,c}+B_{i,c}\zeta_i+C_{i,c}\hat{\zeta}_i+D_{i,c}z_i,\\
	u_i&=E_{i,c}x_{i,c}+F_{i,c}\zeta_i+G_{i,c}\hat{\zeta}_i+H_{i,c}z_i,
	\end{system}
	\end{equation}
	 where $\hat{\zeta}_i$ is defined in \eqref{etahat} with $\eta_i=M_{i,c}x_{i,c}$, and $x_{c,i}\in\mathbb{R}^{n_i}$, 
	such that output synchronization 
	\begin{equation}\label{synch_out}
	\lim\limits_{t\to\infty}(y_i(t)-y_j(t))=0
	\end{equation}
 is achieved for any $N$ and any graph $\mathcal{G}\in\mathbb{G}^N$.
\end{problem}

Next, we consider regulated output synchronization where output of agents converge to a priori given trajectory $y_r$ generated by a so-called exosystem as
\begin{equation}\label{exo}
\begin{system*}{cl}
\dot{x}_r&=A_rx_r, \quad x_r(0)=x_{r0},\\
y_r&=C_rx_r,
\end{system*}
\end{equation}
where $x_r \in\mathbb{R}^r$ and $y_r\in\mathbb{R}^p$.

We assume a nonempty subset $\mathscr{C}$ of the agents which have access to their output relative to the output of the exosystem. In other word, each agent $i$ has access to the quantity 
\begin{equation}
\Psi_i=\iota_i(y_i-y_r), \qquad \iota_i=\begin{system}{cl}
1, \quad i\in \mathscr{C},\\
0, \quad i\notin \mathscr{C}.
\end{system}
\end{equation}
By combining this with \eqref{zeta1}, the information exchange among agents is given by

\begin{equation}\label{zetabar}
\tilde{\zeta}_i=\sum_{j=1}^{N}a_{ij}(y_i-y_j)+\iota_i(y_i-y_r).
\end{equation}
$\tilde{\zeta}_i$ as defined in above, can be rewritten in terms of the coefficients of a so-called expanded Laplacian matrix $\tilde{L}=L+diag\{\iota_i\}=[\tilde{\ell}_{ij}]_{N \times N}$ as
	\begin{equation}\label{zetabar2}
	\tilde{\zeta}_i=\sum_{j=1}^{N}\tilde{\ell}_{ij}(y_j-y_r).
	\end{equation}
	
Note that $\tilde{L}$ is not a regular Laplacian matrix associated to the graph, since the sum of its rows need not be zero. We know that all the eigenvalues of $\tilde{L}$, have positive real parts. In particular matrix $\tilde{L}$ is invertible.

To guarantee that each agent gets the information from the exosystem, we need to make sure that there exists a path from node set $\mathscr{C}$ to each node.  Therefore, we define the following set of graphs.
\begin{definition}\label{def_rootset}
	Given a node set $\mathscr{C}$, we denote by $\mathbb{G}_{\mathscr{C}}^N$ the set of all graphs with $N$ nodes containing the node set $\mathscr{C}$, such that every node of the network graph $\mathcal{G}\in\mathbb{G}_\mathscr{C}^N$ is a member of a directed tree
	which has its root contained in the node set $\mathscr{C}$. We will refer to the node set $\mathscr{C}$ as root set.
\end{definition}

\begin{remark}
	Note that Definition \ref{def_rootset} does not require necessarily the existence of directed spanning tree.
\end{remark}

Now we formulate the problem of \textbf{scalable} regulated output synchronization for a heterogeneous MAS.

\begin{problem}\label{prob_reg_sync}
	Consider a heterogeneous network of $N$ agents \eqref{hete_sys} with local information \eqref{local} satisfying Assumption \ref{ass2} and the associated exosystem \eqref{exo} satisfying Assumption \ref{ass-exo}. Let a set of nodes $\mathscr{C}$ be given which defines the set $\mathbb{G}^N_\mathscr{C}$. Let the associated network communication be given by \eqref{zetabar2}. The \textbf{scalable regulated output synchronization problem based on localized information exchange} is to find, if possible, a linear dynamic controller for each agent $i \in\{1, \dots, N\}$, using only knowledge of agent models, i.e. $(C_i, A_i, B_i)$, of the form:
	\begin{equation}\label{out_reg_dyn}
	\begin{system}{cl}
	\dot{x}_{i,c}&=A_{i,c}x_{i,c}+B_{i,c}\tilde{\zeta}_i+C_{i,c}\hat{\zeta}_i+D_{i,c}z_i,\\
	u_i&=E_{i,c}x_{i,c}+F_{i,c}\tilde{\zeta}_i+G_{i,c}\hat{\zeta}_i+H_{i,c}z_i,
	\end{system}
	\end{equation}
	 where $\hat{\zeta}_i$ is defined in \eqref{etahat} with $\eta_i=M_{i,c}x_{i,c}$, and $x_{c,i}\in\mathbb{R}^{n_i}$, 
	such that regulated output synchronization 
	\begin{equation}\label{reg_synch_out}
	\lim\limits_{t\to\infty}(y_i(t)-y_r(t))=0
	\end{equation}
	 is achieved for any $N$ and any graph $\mathscr{G}\in\mathbb{G}^N_\mathscr{C}$.
\end{problem}

In this paper, we make the following assumptions for agents and the exosystem.
\begin{assumption}\label{ass2}
	For agents $i \in \{1,\dots,N\}$, 
	\begin{enumerate}
		\item $(C_i,A_i,B_i)$ is stabilizable, detectable and right-invertible.
		\item $(C_i^m,A_i)$ is detectable. 
	\end{enumerate}
\end{assumption}

\begin{assumption}\label{ass-exo}
	For exosystem,
	\begin{enumerate}
		\item $(C_r, A_r)$ is observable.
		\item  All the eigenvalues of $A_r$ are on the imaginary axis.
	\end{enumerate}
\end{assumption}

\section{Scalable Output Synchronization}\label{OS}
In this section, we design protocols to solve scalable output synchronization problem as stated in Problem \ref{prob_sync}. After introducing the architecture of our protocol, we design the protocols through four steps.

\subsection{{Architecture of the protocol}}
Our protocol has the structure shown below in Figure \ref{Heterogeneous}.
\begin{figure}[ht]
	\includegraphics[width=8.3cm, height=4.3cm]{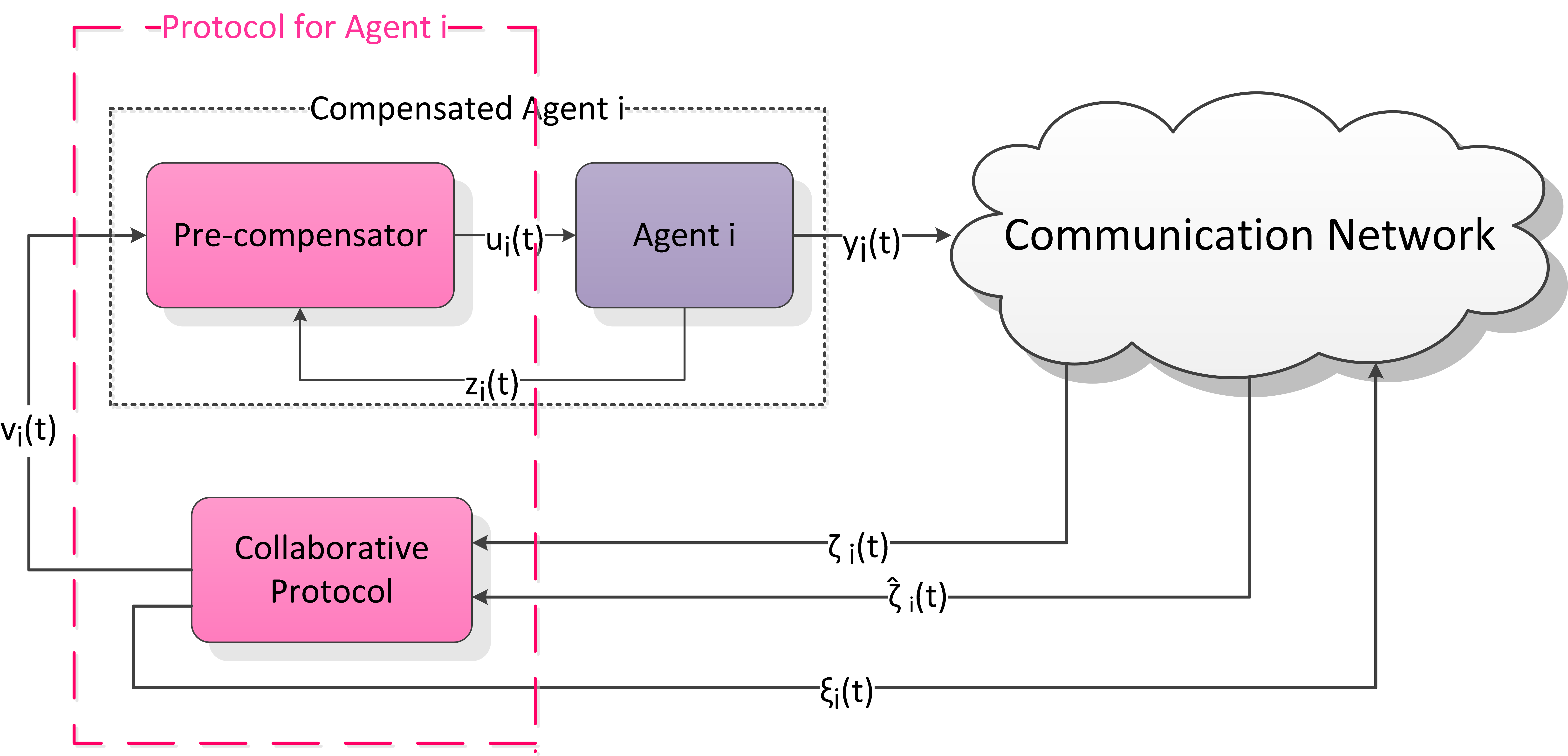}
	\centering
	\caption{Architecture of the protocol for output synchronization }\label{Heterogeneous}
	\vspace*{-.3cm}
\end{figure}

As seen in the above figure, our design methodology consists of two major modules.
\begin{enumerate}
	\item The first module is reshaping the dynamics of the agents to obtain the target model by designing pre-compensators following our previous results in \cite{yang-saberi-stoorvogel-grip-journal}.
	\item The second module is designing collaborate protocols for almost homogenized agents to achieve output and regulated output synchronization
\end{enumerate}

\subsection{Protocol design}
 We design our protocols through the following four steps.
 
\textbf{Step 1: choosing target model} First, we choose the suitable target model, i.e. $(C,A,B)$ such that the following conditions are satisfied.
\begin{enumerate}
	\item $\rank(C)=p$
	\item $(C, A, B)$ is invertible of uniform rank $n_q\ge\bar{n}_d$, and has no invariant zeros, where $\bar{n}_d$ denotes the maximal order of infinite zeros of $(C_i,A_i,B_i), i=1,\ldots,N$.
	\item eigenvalues of $A$ are in closed left half plane.
\end{enumerate}
\textbf{Step 2: designing pre-compensators} In this step, we design pre-compensators to reshape agent models to almost identical agents.  Given chosen target model $(C,A,B)$, by utilizing the design methodology from \cite[Appendix B]{yang-saberi-stoorvogel-grip-journal}, we design a pre-compensator for each agent $i \in \{1, \dots, N\}$, of the form
\begin{equation}\label{pre_con}
\begin{system}{cl}
\dot{\xi}_i&=A_{i,h}\xi_i+B_{i,h}z_i+E_{i,h}v_i,\\
u_i&=C_{i,h}\xi_i+D_{i,h}v_i,
\end{system}
\end{equation} 
which yields the compensated agents as
\begin{equation}\label{sys_homo}
\begin{system*}{cl}
\dot{\bar{x}}_i&=A\bar{x}_i+B(v_i+\rho_i),\\
{y}_i&=C\bar{x}_i,
\end{system*}
\end{equation} 
where $\rho_i$ is given by 
\begin{equation}\label{sys-rho}
\begin{system*}{cl}
\dot{\omega}_i&=A_{i,s}\omega_i,\\
\rho_i&=C_{i,s}\omega_i,
\end{system*}
\end{equation}
and $A_{i,s}$ is Hurwitz stable. Note that the compensated agents are homogenized and have the target model $(C, A, B)$. 

\textbf{Step 3: designing collaborative protocols for compensated agents}
 In this step, we design a dynamic protocol based on localized information exchange for compensated agents \eqref{sys_homo} and \eqref{sys-rho} as
\begin{equation}\label{pscp1}
\begin{system}{cl}
\dot{\hat{x}}_i&=A\hat{x}_i-BK\hat{\zeta}_i+H(\zeta_i-C\hat{x}_i)\\
\dot{\chi}_i&=A\chi_i+Bv_i+\hat{x}_i-\hat{\zeta}_i\\
v_i&=-K\chi_i
\end{system}
\end{equation}
where $H$ and $K$ are matrices such that $A-HC$ and $A-BK$ are Hurwitz stable. The agents communicate $\eta_i$ which is chosen as $\eta_i=\chi_i$. Therefore, each agent has access to the following information:
\begin{equation}\label{hatzeta}
	\hat{\zeta}_i=\sum_{j=1}^{N}a_{ij}(\chi_i-\chi_j),
\end{equation}
 and $\zeta_i$ is defined as \eqref{zeta}.
 
 \textbf{Step 4: obtaining protocols for the agents}
Finally, our protocol which is the combination of module $1$ and $2$ is as following.
\begin{equation}\label{pscp1final}
\begin{system}{cl}
\dot{\xi}_i&=A_{i,h}\xi_i+B_{i,h}z_i-E_{i,h}K\chi_i,\\
\dot{\hat{x}}_i&=A\hat{x}_i-BK\hat{\zeta}_i+H(\zeta_i-C\hat{x}_i)\\
\dot{\chi}_i&=A\chi_i-BK\chi_i+\hat{x}_i-\hat{\zeta}_i\\
u_i&=C_{i,h}\xi_i-D_{i,h}K\chi_i,
\end{system}
\end{equation}

  Then, we have the following theorem for output synchronization of heterogeneous MAS.
 
 \begin{theorem}\label{thm_out_syn}
	Consider a heterogeneous network of $N$ agents \eqref{hete_sys} with local information \eqref{local} satisfying Assumption \ref{ass2}. Let the associated network communication be given by \eqref{zeta}. Then, the scalable output synchronization problem as defined in Problem \ref{prob_sync} is solvable.  In particular, the dynamic protocol \eqref{pscp1final} with localized information \eqref{hatzeta} solves the scalable output synchronization problem based on localized information exchange for any $N$ and any graph
 	$\mathcal{G}\in\mathbb{G}^N$. 
 \end{theorem} 

\begin{remark}
	It is interesting to note that for the case that agents are \textbf{homogeneous} and \textbf{non-introspective} one does not require designing pre-compensators, and obviously can achieve scalable output synchronization utilizing collaborative protocols proposed in \textit{step $3$} for homogeneous networks as long as the agents eigenvalues are in the closed left half plane.
\end{remark}	

To obtain this result, we recall the following lemma for Laplacian matrix $L$. 
\begin{lemma}[\cite{zhang-saberi-stoorvogel-delay}]\label{LbarL}\label{lemmaLbar}
	Let a Laplacian matrix $L\in \R^{N\times N}$ be given associated
	with a graph that contains a directed spanning tree. We define
	$\bar{L}\in \R^{(N-1)\times (N-1)}$ as the matrix
	$\bar{L}=[\bar{\ell}_{ij}]$ with
	\[
	\bar{\ell}_{ij} = \ell_{ij}-\ell_{Nj}.
	\]
	Then the eigenvalues of $\bar{L}$ are equal to the nonzero
	eigenvalues of $L$.
\end{lemma}

\begin{proof}
	We have:
	\[
	\bar{L} = \begin{pmatrix} I & -\textbf{1} \end{pmatrix}
	L \begin{pmatrix} I \\ 0 \end{pmatrix}
	\]
	Assume that $\lambda$ is a nonzero eigenvalue of $L$ with eigenvector
	$x$, Then
	\[
	\bar{x} = \begin{pmatrix} I & -\textbf{1} \end{pmatrix} x
	\]
	where $\textbf{1}$ is a vector with all $1$'s, satisfies,
	\[
	\begin{pmatrix} I & -\textbf{1} \end{pmatrix} Lx =
	\begin{pmatrix} I & -\textbf{1} \end{pmatrix} \lambda x
	=\lambda \bar{x}
	\]
	and since $L\textbf{1}=0$ we find that 
	\[
	\bar{L} \bar{x} =
	\begin{pmatrix} I & -\textbf{1} \end{pmatrix} Lx =\lambda \bar{x}.
	\]
	This shows that  $\lambda$ is an eigenvector of $\bar{L}$ if
	$\bar{x}\neq0$. It is easily seen that $\bar{x}=0$ if and only if
	$\lambda=0$. Conversely if $\bar{x}$ is an eigenvector of $\bar{L}$
	with eigenvalue $\lambda$ then it is easily verified that
	\[
	x = L \begin{pmatrix} I \\ 0 \end{pmatrix} \bar{x}
	\]
	is an eigenvector of $L$ with eigenvalue $\lambda$.
\end{proof}

\begin{proof}[Proof of Theorem \ref{thm_out_syn}] Let
	$\bar{x}_i^o=\bar{x}_i-\bar{x}_N$, $y_i^o=y_i-y_N$, $\hat{x}_i^o=\hat{x}_i-\hat{x}_N$, and $\chi_i^o=\chi_i-\chi_N$. Then, we have 
	\[
	\begin{system*}{ll}
	\dot{\bar{x}}_i^o&=A\bar{x}_i^o+B(v_i-v_N+\rho_i-\rho_N),\\
	{y}_i^o&=C\bar{x}_i^o,\\
	\bar{\zeta}_i&=\zeta_i-\zeta_N=\sum_{j=1}^{N-1}\bar{\ell}_{ij}{y}_j^o,\\
	\dot{\hat{x}}_i^o&=A\hat{x}_i^o-BK(\hat{\zeta}_i-\hat{\zeta}_N)+H(\bar{\zeta}_i-C\hat{x}_i^o)\\
	\dot{\chi}_i^o&=A\chi_i^o+B(v_i-v_N)+\hat{x}_i^o-\sum_{j=1}^{N-1}\bar{\ell}_{ij}{\chi}_j^o\\
	v_i-v_N&=-K\chi_i^o
	\end{system*}
	\]
	where $\bar{\ell}_{ij}=\ell_{ij}-\ell_{Nj}$ for $i,j=1,\cdots,N-1$. According to Lemma \ref{lemmaLbar}, we have that eigenvalues of $\bar{L}=[\bar{\ell}_{ij}]_{(N-1)\times(N-1)}$ are equal to the nonzero eigenvalues of $L$.
We define
\begin{equation*}
	\bar{x}=\begin{pmatrix}
	\bar{x}_1^o\\ \vdots\\ \bar{x}_{N-1}^o
	\end{pmatrix},\hat{x}=\begin{pmatrix}
	\hat{x}_1^o\\ \vdots\\ \hat{x}_{N-1}^o
	\end{pmatrix},\chi=\begin{pmatrix}
	\chi_1^o\\ \vdots\\ \chi_{N-1}^o
	\end{pmatrix},\rho=\begin{pmatrix}
	\rho_1\\ \vdots\\ \rho_N\end{pmatrix},\omega=\begin{pmatrix}
	\omega_1\\ \vdots\\ \omega_N\end{pmatrix}
\end{equation*} 
then, we have the following closed-loop system:
\begin{equation}
	\begin{system}{cl}
	\dot{\bar{x}}&=(I\otimes A)\bar{x}-(I\otimes BK )\chi+(\Pi\otimes B)\rho\\
	\dot{\hat{x}}&=(I\otimes (A-HC))\hat{x}-(\bar{L}\otimes BK )\chi+(\bar{L}\otimes HC)\bar{x}\\
	\dot{\chi}&=(I\otimes(A-BK )-\bar{L}\otimes I)\chi+\hat{x}
	\end{system}
\end{equation}
with $\Pi=\begin{pmatrix}
I&-\mathbf{1}
\end{pmatrix}$.
By defining $e=\bar{x}-\chi$ and $\bar{e}=(\bar{L}\otimes I)\bar{x}-\hat{x}$, we can obtain
\begin{equation}\label{x-e}
\begin{system*}{cl}
	\dot{\bar{x}}&=(I\otimes (A-BK ))\bar{x}+(I\otimes BK )e+(\Pi\otimes B) C_s\omega\\
	\dot{e}&=(I\otimes A-\bar{L}\otimes I)e+\bar{e}+(\Pi\otimes B) C_s\omega\\
	\dot{\bar{e}}&=(I\otimes(A-HC))\bar{e}+(\bar{L}\Pi\otimes B) C_s\omega
	\end{system*}
\end{equation}
where $C_s=diag\{C_{i,s}\}$ for $i=1,\hdots ,N$. We obtain the synchronization result when $\lim_{t \to \infty}\bar{x}\to 0.$ By combining \eqref{sys-rho} and \eqref{x-e}, we have
\begin{multline}\label{bar_x}
\begin{pmatrix}
\dot{\bar{x}}\\\dot{e}\\\dot{\bar{e}}\\\dot{\omega}
\end{pmatrix}=\left(\begin{array}{cc}
I\otimes (A-BK )&I\otimes BK \\
0&I_{N-1}\otimes A-\bar{L}\otimes I\\
0&0\\
0&0
\end{array}\right.
\\
\left.\begin{array}{cc}
0&(\Pi\otimes B) C_s\\
I&(\Pi\otimes B) C_s\\
I\otimes(A-HC)&(\bar{L}\Pi\otimes B) C_s\\
0&A_s
\end{array}\right)
\begin{pmatrix}
\bar{x}\\e\\\bar{e}\\\omega
\end{pmatrix}
\end{multline}
where $A_s=diag\{A_{i,s}\}$ for $i=1,\hdots ,N$.   Since the eigenvalues $\lambda_1,\ldots, \lambda_N$ of $\bar{L}$ have positive real part, we have, we have
\begin{equation}\label{boundapl}
(T\otimes I)(I\otimes A-\bar{L}\otimes I)(T^{-1}\otimes I)=I\otimes A-\bar{\Lambda}\otimes I
\end{equation}
for a non-singular transformation matrix $T$, where
 \eqref{boundapl}  is upper triangular Jordan form with $A-\lambda_i I$ for $i=1,\cdots,N-1$ on the diagonal. Since all eigenvalues of $A$ are in the closed left half plane, $A-\lambda_i I$ is stable. Therefore, all eigenvalues of $I\otimes A-\bar{L}\otimes I$ have negative real part. Moreover, $A_s=\diag(A_{i,s})$ and $A-HC$ are Hurwitz stable, thus it implies we just need to prove the stability of $\dot{\bar{x}}=(I\otimes (A-BK ))\bar{x}$ based on the structure of \eqref{bar_x}. Obviously, since $A-BK$ is Hurwitz stable, we have $\bar{x}$ is asymptotically stable, i.e. $\lim_{t \to \infty}\bar{x}_i^o\to 0$. Thus, it implies $\lim_{t \to \infty}y_i^o=y_i-y_N=C\bar{x}_i^o\to 0$ which proves the result.
\end{proof}

\section{Scalable Regulated Output Synchronization}
\subsection{Architecture of the protocol}
The protocol for regulated output synchronization has two main modules as shown in figure \ref{Heterogeneous_reg}.
\begin{figure}[h]
	\includegraphics[width=8.3cm, height=4.3cm]{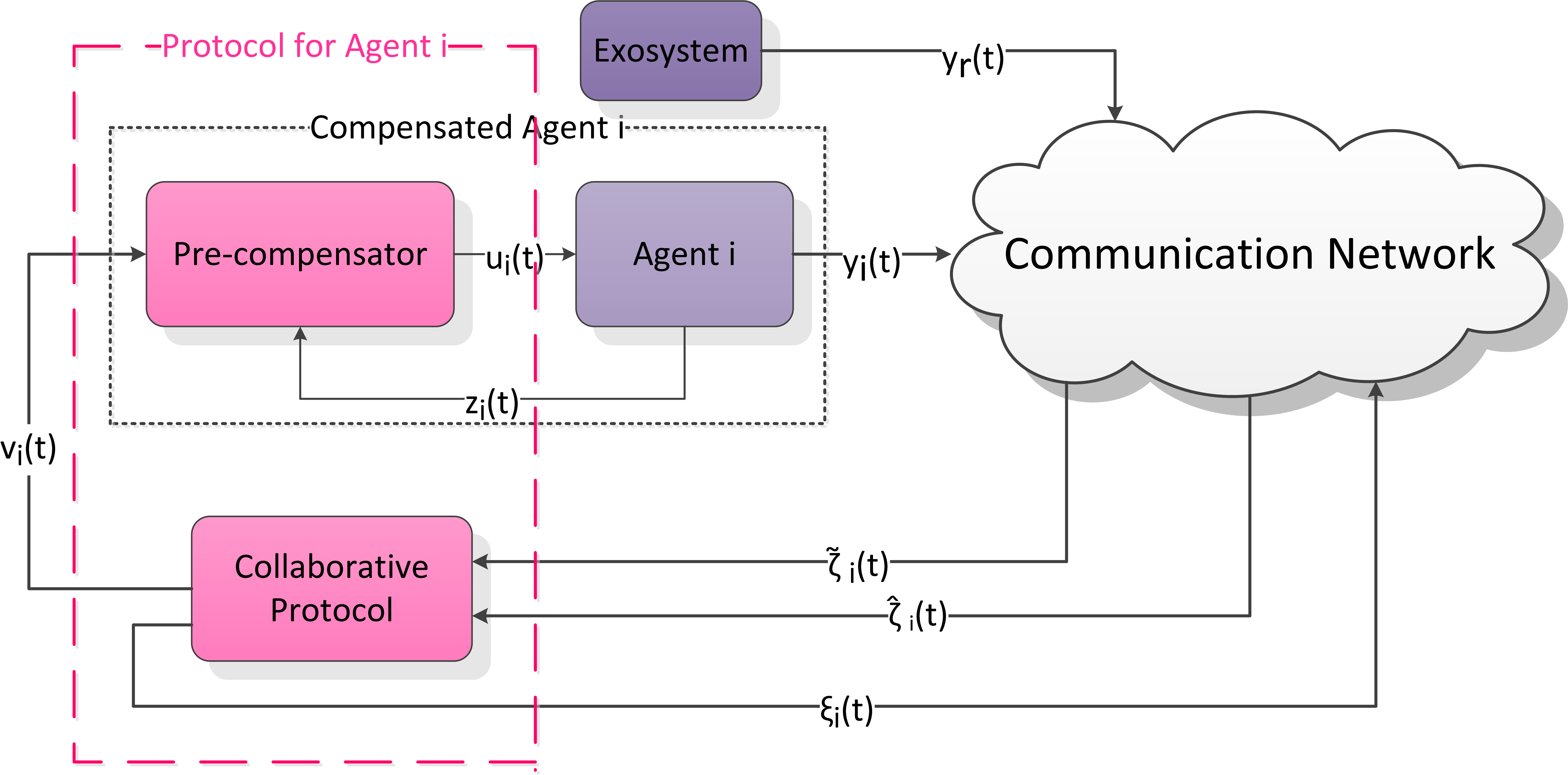}
	\centering
	\caption{Architecture of the protocol for regulated output synchronization}\label{Heterogeneous_reg}
	\vspace*{-.4cm}
\end{figure}
\subsection{Protocol design}
Similar to scalable output synchronization, for solving scalable regulated output synchronization, our design procedure consists of four steps. 

\textbf{Step 1: remodeling the exosystem} First, we remodel the exosystem to arrive at suitable choice for the target model $(\check{C}_r,\check{A}_r,\check{B}_r)$ following the design procedure in \cite{yang-saberi-stoorvogel-grip-journal} summarized in the following lemma.
\begin{lemma}[\cite{yang-saberi-stoorvogel-grip-journal}]\label{lem-exo} There exists another exosystem given by:
	\begin{equation}\label{exo-2}
	\begin{system*}{cl}
	\dot{\check{x}}_r&=\check{A}_r\check{x}_r, \quad \check{x}_r(0)=\check{x}_{r0}\\
	y_r&=\check{C}_r\check{x}_r,
	\end{system*}
	\end{equation}
	such that for all $x_{r0} \in \mathbb{R}^r$, there exists $\check{x}_{r0}\in \mathbb{R}^{\check{r}}$ for which \eqref{exo-2} generate exactly the same output $y_r$ as the original exosystem \eqref{exo}. Furthermore, we can find a matrix $\check{B}_r$ such that the triple $(\check{C}_r,\check{A}_r,\check{B}_r)$ is invertible, of uniform rank $n_q$, and has no invariant zero, where $n_q$ is an integer greater than or equal to maximal order of infinite zeros of $(C_i,A_i,B_i), i\in \{1,...,N\}$ and all the observability indices of $(C_r, A_r)$. Note that the eigenvalues of $\check{A}_r$ consists of all eigenvalues of $A_r$ and additional zero eigenvalues. 
\end{lemma}

 \textbf{Step 2: designing pre-compensators} Next, given the target model $(\check{C}_r,\check{A}_r,\check{B}_r)$, we design pre-compensators to achieve almost identical models as in \textit{step $2$} of section \ref{OS}. 
 
  \textbf{Step 3: designing collaborative protocols for the compensated agents} Collaborative protocols based on localized information exchanges are designed for the compensated agents $i=1,\hdots, N$ as
\begin{equation}\label{pscp2}
\begin{system}{cl}
\dot{\hat{x}}_i&=\check{A}_r\hat{x}_i-\check{B}_rK\hat{\zeta}_i+H(\tilde{\zeta}_i-\check{C}_r\hat{x}_i)+\iota_i \check{B}_r v_i,\\
\dot{\chi}_i&=\check{A}_r\chi_i+\check{B}_rv_i+\hat{x}_i-\hat{\zeta}_i-\iota_i\chi_i,\\
v_i&=-K\chi_i,
\end{system}
\end{equation}
where $H$ and $K$ are design matrices such that $\check{A}_r-H\check{C}_r$ and $\check{A}_r-\check{B}_rK$ are Hurwitz stable.
 The exchanging information $\hat{\zeta}_i$ is defined as \eqref{hatzeta} and $\tilde{\zeta}_i$ is defined as \eqref{zetabar2}. 
 
   \textbf{Step 4: obtaining the protocols}
The final protocol which is the combination of module $1$ and $2$ is
 	\begin{equation}\label{pscp2final}
 \begin{system}{cl}
 \dot{\xi}_i&=A_{i,h}\xi_i+B_{i,h}z_i-E_{i,h}K\chi_i,\\
\dot{\hat{x}}_i&=\check{A}_r\hat{x}_i-\check{B}_rK\hat{\zeta}_i+H(\tilde{\zeta}_i-\check{C}_r\hat{x}_i)-\iota_i \check{B}_r K\chi_i,\\
\dot{\chi}_i&=\check{A}_r\chi_i-\check{B}_rK\chi_i+\hat{x}_i-\hat{\zeta}_i-\iota_i\chi_i,\\
 u_i&=C_{i,h}\xi_i-D_{i,h}K\chi_i,
 \end{system}
 \end{equation} 
 
 Then, we have the following theorem for scalable regulated output synchronization of heterogeneous MAS.
  
\begin{theorem}\label{thm_reg_out_syn}
 	Consider a heterogeneous network of $N$ agents \eqref{hete_sys} satisfying Assumption \ref{ass2} with local information \eqref{local} and the associated exosystem \eqref{exo} satisfying Assumption \ref{ass-exo}. Then, the scalable regulated output synchronization problem as defined in Problem \ref{prob_reg_sync} is solvable.  In particular, the dynamic protocol \eqref{pscp2final} solves the scalable regulated output synchronization problem based on localized information exchange for any $N$ and any graph
	$\mathscr{G}\in\mathbb{G}^N_\mathscr{C}$. 
\end{theorem}

\begin{proof}[Proof of Theorem \ref{thm_reg_out_syn}] Similar to Theorem \ref{thm_out_syn}, we design a pre-compensator \eqref{pre_con} for each agent $i \in \{1,..., N\}$ to obtain our target model as
	
	\begin{equation}\label{sys_reg_homo}
	\begin{system*}{cl}
	\dot{\bar{x}}_i&=\check{A}_r\bar{x}_i+\check{B}_r(v_i+\rho_i),\\
	{y}_i&=\check{C}_r\bar{x}_i,
	\end{system*}
	\end{equation}
	where $\rho_i$ is given by \eqref{sys-rho}. Let $\tilde{x}_i=\bar{x}_i-\check{x}_r$ and define
	\begin{equation*}
	\tilde{x}=\begin{pmatrix}
	\tilde{x}_1\\ \vdots\\ \tilde{x}_N
	\end{pmatrix},\hat{x}=\begin{pmatrix}
	\hat{x}_1\\ \vdots\\ \hat{x}_N
	\end{pmatrix},\chi=\begin{pmatrix}
	\chi_1\\ \vdots\\ \chi_N
	\end{pmatrix},\rho=\begin{pmatrix}
	\rho_1\\ \vdots\\ \rho_N\end{pmatrix},\omega=\begin{pmatrix}
	\omega_1\\ \vdots\\ \omega_N\end{pmatrix}
	\end{equation*}
	then, we have the following closed-loop system
	\begin{equation}
	\begin{system}{cl}
\dot{\tilde{x}}&=(I\otimes \check{A}_r)\tilde{x}-(I\otimes \check{B}_rK)\chi+(I\otimes \check{B}_r)\rho\\
\dot{\hat{x}}&=(I\otimes (\check{A}_r-H\check{C}_r))\hat{x}-(\tilde{L}\otimes \check{B}_rK)\chi+(\tilde{L}\otimes H\check{C}_r)\tilde{x}\\
\dot{\chi}&=(I\otimes(\check{A}_r-\check{B}_rK)-\tilde{L}\otimes I)\chi+\hat{x}
\end{system}
	\end{equation}
		
	By defining $e=\tilde{x}-\chi$ and $\bar{e}=(\tilde{L}\otimes I)\tilde{x}-\hat{x}$, we can obtain  
	\begin{equation}\label{newsystem3}
\begin{system*}{cl}
\dot{\tilde{x}}&=(I\otimes (\check{A}_r-\check{B}_rK ))\tilde{x}+(I\otimes \check{B}_rK )e+(I\otimes \check{B}_r)C_s\omega\\
\dot{e}&=(I\otimes \check{A}_r-\tilde{L}\otimes I)e+\bar{e}+(I\otimes \check{B}_r)C_s\omega\\
\dot{\bar{e}}&=(I\otimes(\check{A}_r-H\check{C}_r))\bar{e}+(\tilde{L}\otimes \check{B}_r)C_s\omega
\end{system*}
	\end{equation}
	
	Similar to Theorem \ref{thm_out_syn}, since all eigenvalues of $\tilde{L}$ have positive real part, we obtain that $e$ and $\bar{e}$ have asymptotically stable dynamics. Therefore, we just need to prove the stability of 
	\begin{equation}\label{statefeedback2}
	\dot{\tilde{x}}=(I\otimes (\check{A}_r-\check{B}_rK ))\tilde{x}
	\end{equation}
	
	Thus, according to the result of Theorem \ref{thm_out_syn}, 
	we can obtain the asymptotic stability of \eqref{statefeedback2}, i.e., $\lim_{t\to \infty}\tilde{x}_i\to 0$. It implies that $\bar{x}_i-\check{x}_r\to0$ which proves the result.
\end{proof}

\section{Simulation Results}
In this section, we will illustrate the effectiveness of our protocols with a numerical example for output synchronization of heterogeneous MAS with partial-state coupling. We show that our protocol design \eqref{pscp1final} is scale-free and it works for any graph with any number of agents.  
\begin{figure}[t]
	\includegraphics[width=4cm, height=3.5cm]{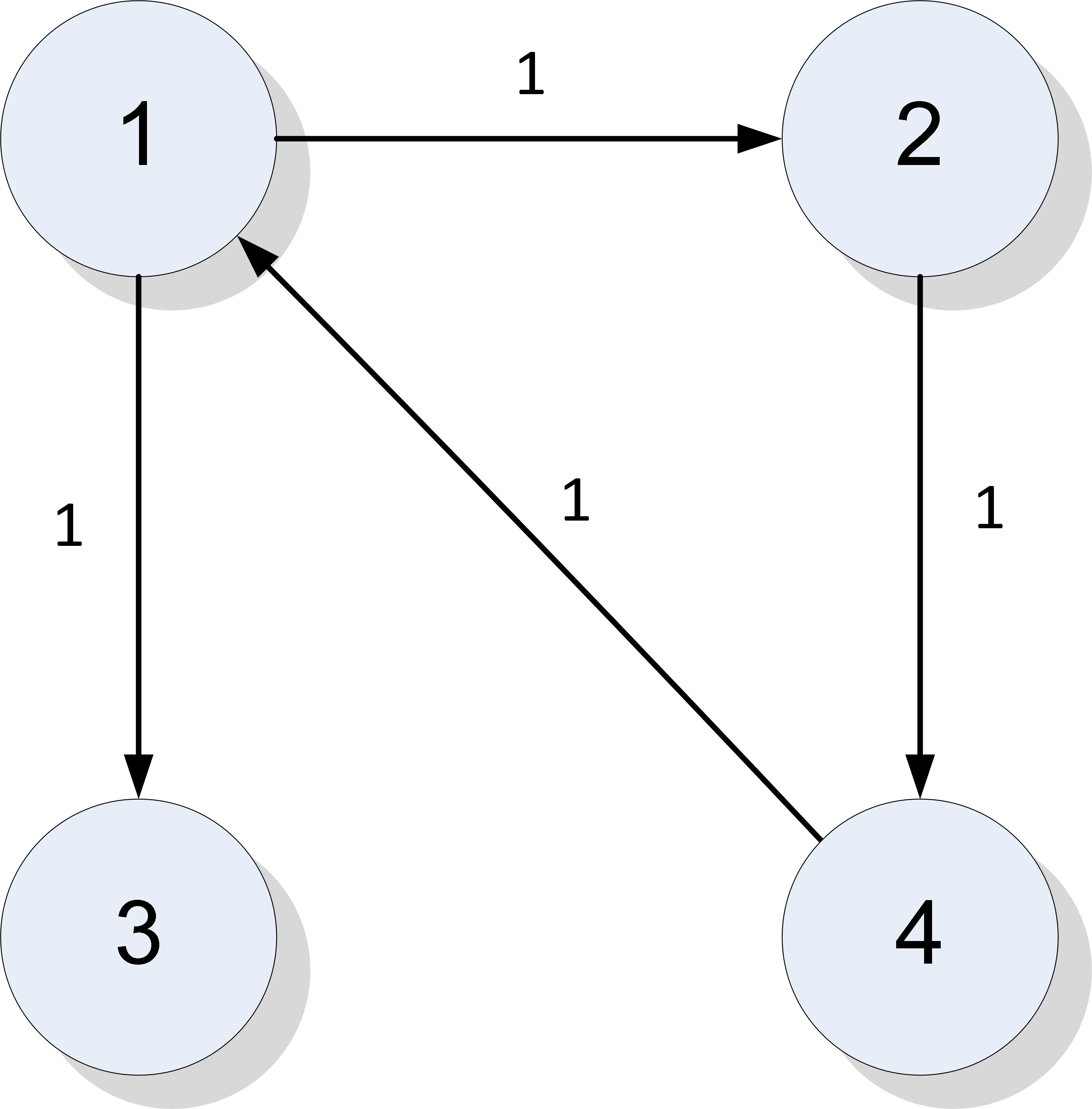}
	\centering
	\caption{Communication network topology for case $1$}\label{Graph_1}
\end{figure}
\begin{figure}[t]
	\includegraphics[width=4cm, height=3.7cm]{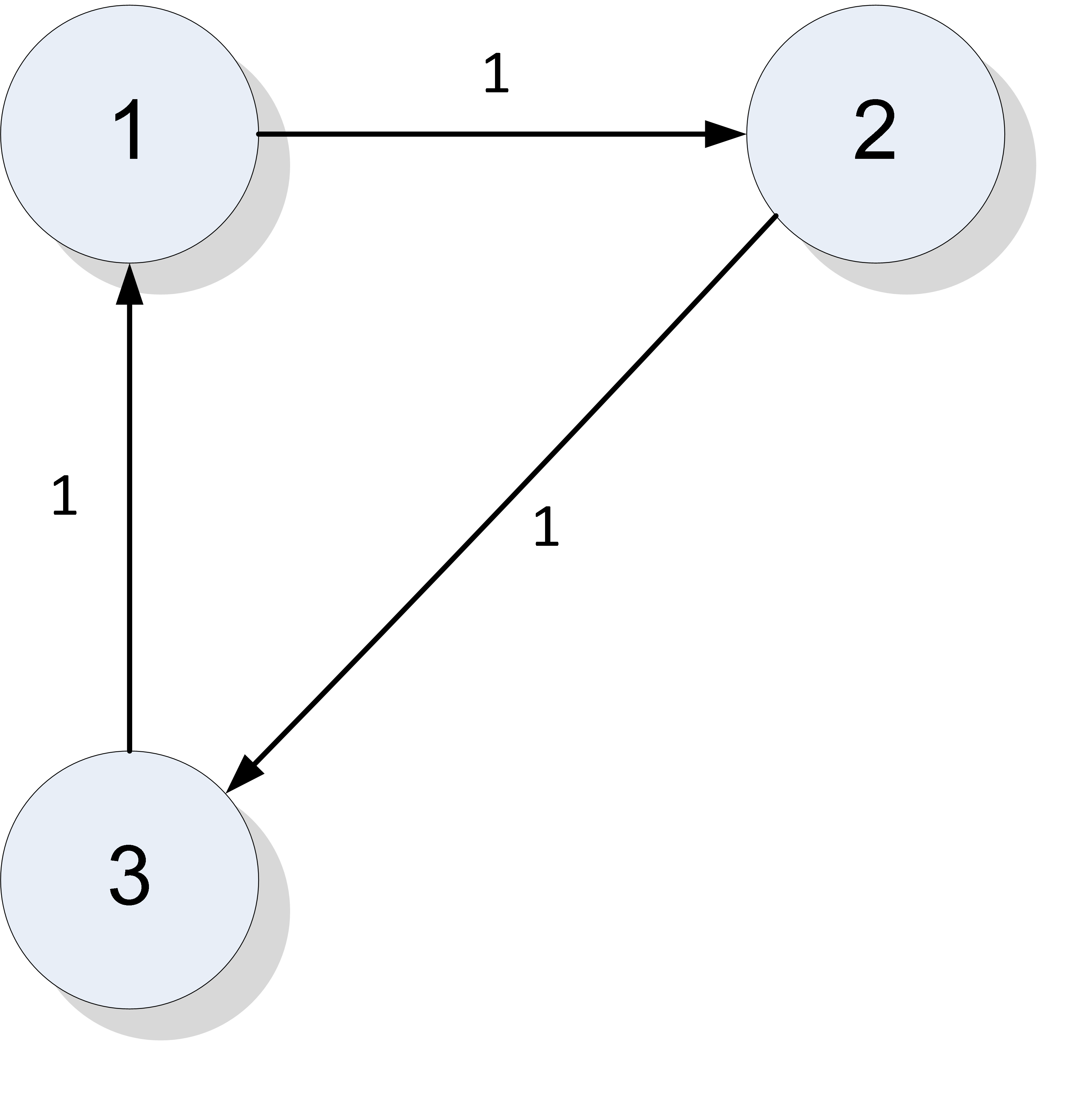}
	\centering
	\caption{Communication network topology for case $2$}\label{Graph_2}
\end{figure}
\begin{figure}[t!]
	\includegraphics[width=6cm, height=3.5cm]{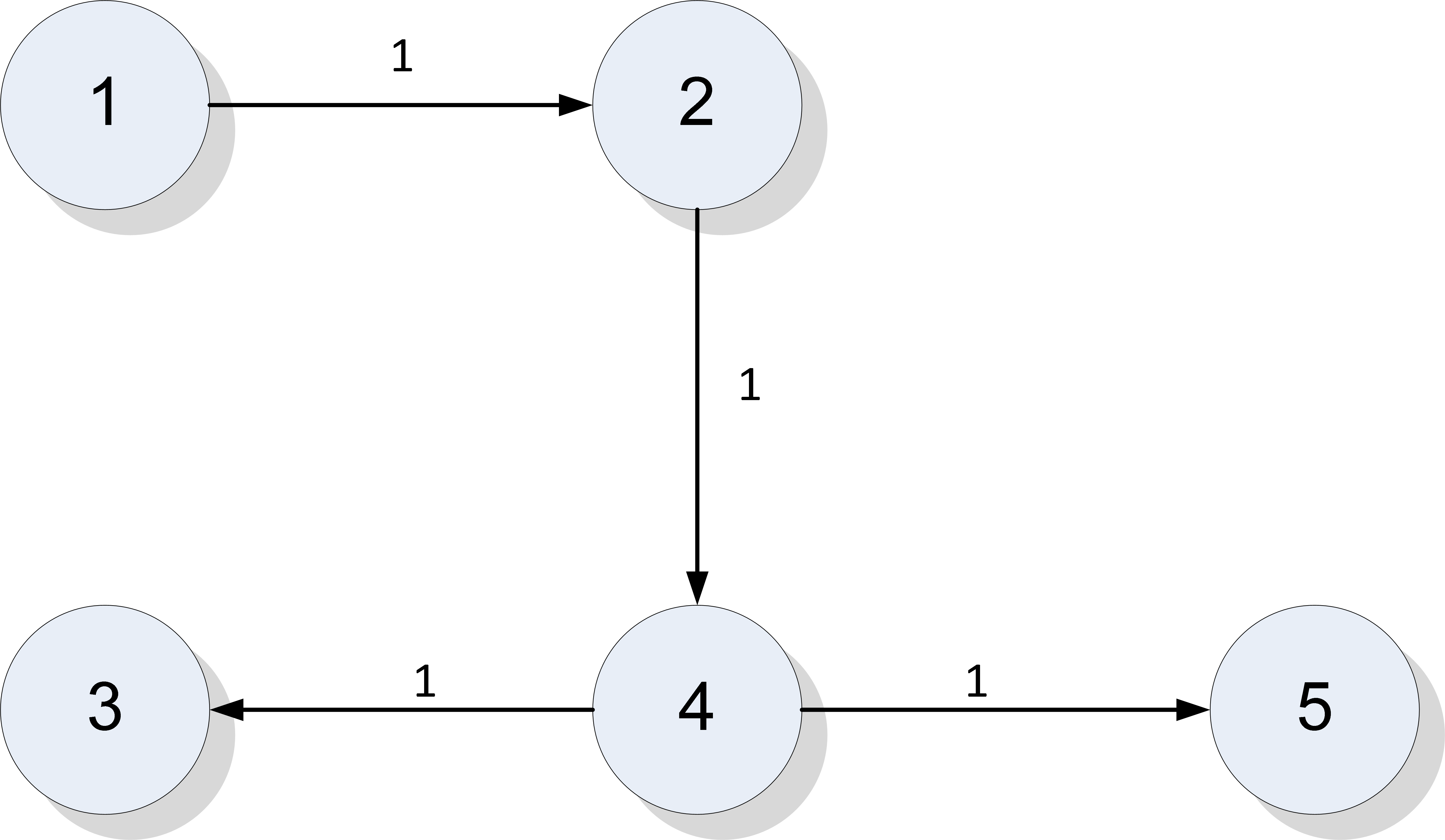}
	\centering
	\caption{Communication network topology for case $3$}\label{Graph_3}
\end{figure}
Consider the agents models \eqref{hete_sys} as:
\begin{equation*}
\begin{system*}{cl}
A_1&=\begin{pmatrix}
0&1&0&0\\0&0&1&0\\0&0&0&1\\0&0&0&0
\end{pmatrix},\quad B_1=\begin{pmatrix}
0&1\\0&0\\1&0\\0&1
\end{pmatrix},\\ C_1&=\begin{pmatrix}
1&0&0&0
\end{pmatrix},\quad
C^m_1=I\\
A_2&=\begin{pmatrix}
0&1&0\\0&0&1\\0&0&0
\end{pmatrix},\quad B_2=\begin{pmatrix}
0\\0\\1
\end{pmatrix},\\ C_2&=\begin{pmatrix}
1&0&0
\end{pmatrix},\quad
C^m_2=I,
\end{system*}
\end{equation*}
and for $i=3,4$

\begin{equation*}
\begin{system*}{cl}
A_i&=\begin{pmatrix}
-1&0&0&-1&0\\0&0&1&1&0\\0&1&-1&1&0\\0&0&0&1&1\\-1&1&0&1&1
\end{pmatrix},\quad B_i=\begin{pmatrix}
0&0\\0&0\\0&1\\0&0\\1&0
\end{pmatrix},\\ C_i&=\begin{pmatrix}
0&0&0&1&0
\end{pmatrix},\quad
C^m_i=I,
\end{system*}
\end{equation*}
and 
\begin{equation*}
\begin{system*}{cl}
A_5&=\begin{pmatrix}
0&1&0\\0&0&1\\1&1&0
\end{pmatrix},\quad B_5=\begin{pmatrix}
0\\0\\1
\end{pmatrix},\\ C_5&=\begin{pmatrix}
1&0&0
\end{pmatrix},\quad
C^m_2=I,
\end{system*}
\end{equation*}

Note that $\bar{n}_d=3$, which is the degree of infinite zeros of $(C_2,A_2,B_2)$. We choose $n_q=3$ and matrices $A,B,C$ as following.
\begin{equation*}
\begin{system*}{cl}
A&=\begin{pmatrix}
0&1&0\\0&0&1\\0&-1&0
\end{pmatrix},\quad B=\begin{pmatrix}
0\\0\\1
\end{pmatrix}, \quad C=\begin{pmatrix}
1&0&0
\end{pmatrix}\\
\end{system*}
\end{equation*}
and $K$ and $H$ as:
\[
K=\begin{pmatrix}
30\\30\\10
\end{pmatrix}
\quad H=\begin{pmatrix}
6&10&0
\end{pmatrix}
\]

We consider three different heterogeneous MAS with different number of agents and different communication topologies to show that the designed protocols are independent of the
communication networks and the number of agents $N$.\\

\begin{itemize}
	\item \emph{Case $1$:} Consider a MAS with $4$ agents with agent models $(C_i, A_i, B_i)$ for $i \in \{1,\hdots,4\}$, and directed communication topology shown in Figure \ref{Graph_1}.\\

	\item \emph{Case $2$:} In this case, we consider a MAS with $3$ agents with agent models $(C_i, A_i, B_i)$ for $i \in \{1,\hdots,3\}$ and directed communication topology shown in Figure \ref{Graph_2}.
	\\
	
	\item \emph{Case $3$:} 
	Finally, we consider a MAS with $5$ agents with agent models $(C_i, A_i, B_i)$ for $i \in \{1,\hdots,5\}$ and directed communication topology shown in Figure \ref{Graph_3}.\\
	
\end{itemize}
\begin{figure}[t]
	\includegraphics[width=9cm, height=5cm]{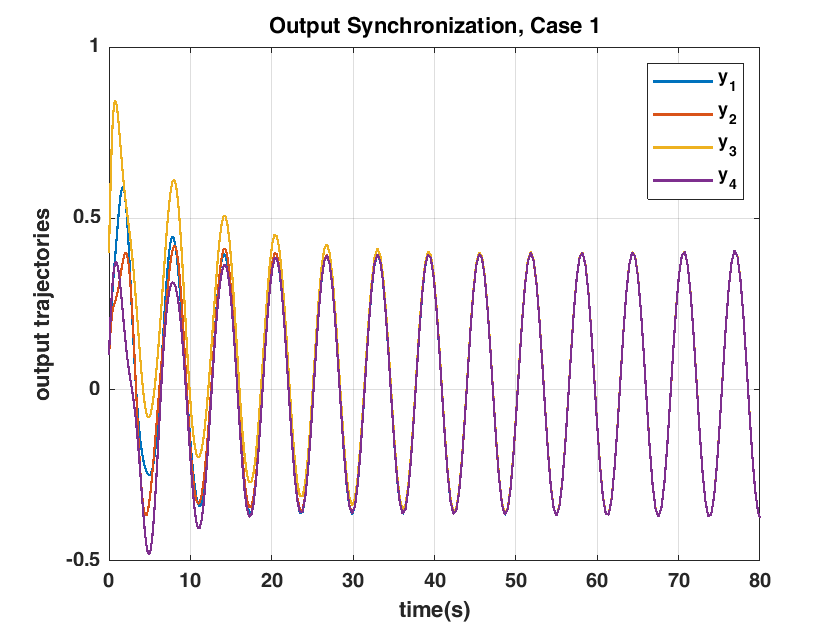}
	\centering
	\caption{Output synchronization for communication networks of case $1$}\label{Results_case1}
\end{figure} 
\begin{figure}[t]
	\includegraphics[width=9cm, height=5cm]{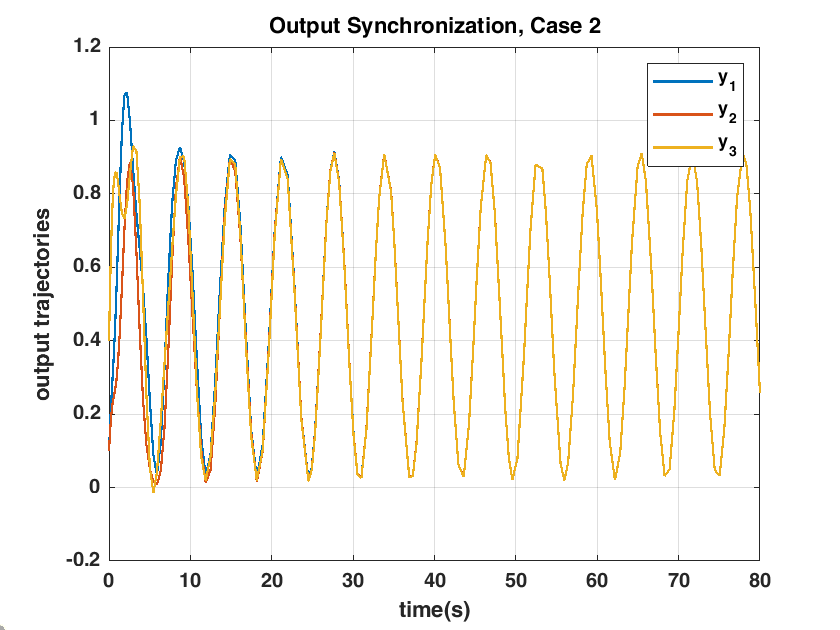}
	\centering
	\caption{Output synchronization for communication networks of case $2$}\label{Results_case2}
\end{figure} 
\begin{figure}[t]
	\includegraphics[width=9cm, height=5cm]{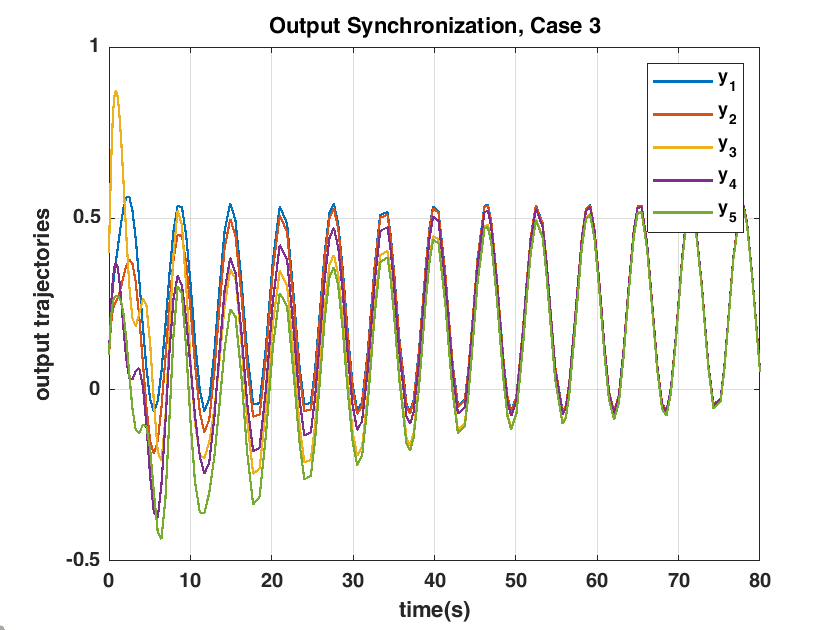}
	\centering
	\caption{Output synchronization for communication networks of case $3$}\label{Results_case3}
\end{figure} 

Figures \ref{Results_case1}-\ref{Results_case3} show that output synchronization is achieved for all three cases. The simulation results also confirm that the protocol design is independent of the communication graph and is scale-free so that we can achieve output synchronization with one-shot protocol design, for any graph with any number of agents.\\

\bibliographystyle{plain}
\bibliography{referenc}

\end{document}